\documentclass[a4paper,english]{lipics-v2018}
\usepackage{cite}
\usepackage{microtype}
\usepackage{todonotes}

\nolinenumbers

\EventEditors{Hiro Ito, Stefano Leonardi, Linda Pagli, and Giuseppe Prencipe}
\EventNoEds{4}
\EventLongTitle{9th International Conference on Fun with Algorithms (FUN 2018)}
\EventShortTitle{FUN 2018}
\EventAcronym{FUN}
\EventYear{2018}
\EventDate{June 13--15, 2018}
\EventLocation{La Maddalena, Italy}
\EventLogo{}
\SeriesVolume{100}
\ArticleNo{20} 

\theoremstyle{definition}

\begin{document}
\title{Faster Evaluation of Subtraction Games}

\author{David Eppstein}{Computer Science Department, University of California, Irvine}{eppstein@uci.edu}{}{Supported in part by NSF grants  CCF-1618301 and CCF-1616248.}
\authorrunning{David Eppstein}
\Copyright{David Eppstein}

\subjclass{}

\maketitle

\begin{abstract}
Subtraction games are played with one or more heaps of tokens, with players taking turns removing from a single heap a number of tokens belonging to a specified \emph{subtraction set}; the last player to move wins. We describe how to compute the set of winning heap sizes in single-heap subtraction games (for an input consisting of the subtraction set and maximum heap size $n$), in time $\tilde O(n)$, where the $\tilde O$ elides logarithmic factors. For multi-heap games, the optimal game play is determined by the \emph{nim-value} of each heap; we describe how to compute the nim-values of all heaps of size up to~$n$ in time $\tilde O(mn)$, where $m$ is the maximum nim-value occurring among these heap sizes.
These time bounds improve naive dynamic programming algorithms with time $O(n|S|)$, because $m\le|S|$ for all such games. We apply these results to the game of subtract-a-square, whose set of winning positions is a maximal square-difference-free set of a type studied in number theory in connection with the Furstenberg--S\'ark\"ozy theorem. We provide experimental evidence that, for this game, the set of winning positions has a density comparable to that of the densest known square-difference-free sets, and has a modular structure related to the known constructions for these dense sets. Additionally, this game's nim-values are (experimentally) significantly smaller than the size of its subtraction set, implying that our algorithm achieves a polynomial speedup over dynamic programming.
\end{abstract}

\section{Introduction}

\emph{Subtraction games} were made famous by the French film \emph{L'Ann\'ee derni\`ere \`a Marienbad} (1961), which showed repeated scenes of two men playing Nim.
A subtraction game is played by two players, with some heaps of game tokens (such as coins, stones, or, in the film, matchsticks) between them. On each turn, a player may take away a number of tokens from a single heap. The tokens removed in each turn are discarded, and play continues until all the tokens are gone. Under the \emph{normal winning convention}, the last player to move is the winner~\cite{BerConGuy-SG-82}. In Nim, any number of tokens may be removed in a turn. This game has a simple analysis according to which it is a winning move to make the bitwise exclusive-or of the binary representations of the heap sizes become zero. If this bitwise exclusive-or is already zero, the player who just moved already has a winning position~\cite{ONAG-11}. However, other subtraction games require the number of removed tokens to belong to a predetermined set of numbers, the \emph{subtraction set} of the game. Different subtraction sets lead to different games with different strategies.\footnote{Golomb~\cite{Gol-JCT-66} has considered an even more general class of games, in which the subtraction set specifies combinations of numbers of tokens that may be simultaneously removed from each pile.}

All subtraction games are \emph{impartial}, meaning that the choice of moves on each turn does not depend on who is making the move. As such, with the normal winning convention, these games can be analyzed by the Sprague--Grundy theory, according to which each heap of tokens in a subtraction game has a \emph{nim-value}, the size of an equivalent heap in the game of Nim~\cite{Spr-Tohoku-35,Gru-Eur-39,ONAG-11}. The optimal play in any such game is to move to make the bitwise exclusive-or of the nim-values zero. The winning positions are the ones in which this bitwise exclusive-or is already zero. Unlike Nim itself, positions with a single nonempty heap of tokens may be winning for the player who just moved; this is true when the nim-value of the heap is zero. The heap sizes whose nim-values are zero are called ``cold'', while the remaining heap sizes are called ``hot''. In a game with a single heap of tokens, it is a winning move to take a number of tokens such that the remaining tokens form a cold position. If the position is already cold, the player who just moved already has a winning position, because the player to move must move to a hot position.

Every finite subtraction set leads to a game with periodic nim-values (depending only on the sizes of the heaps modulo a fixed number). Some natural choices of infinite subtraction set, such as the prime numbers, also do not lead to interesting subtraction games~\cite{Gol-JCT-66}. 
However, a more complicated subtraction game, ``subtract-a-square'', has the square numbers as its subtraction set. That is, on each move, each player  may remove any square number of tokens from any single heap of tokens. The game of subtract-a-square was studied in 1966 by Golomb~\cite{Gol-JCT-66}, who calls it ``remarkably complex''; Golomb credits its invention to Richard A. Epstein.\footnote{No relation.} Its sequence of nim-values,
\[
0, 1, 0, 1, 2, 0, 1, 0, 1, 2, 0, 1, 0, 1, 2, 0, 1, 0, 1, 2, 0, 1, 0, 1, 2, 3, 2, 3, 4, 5, 3, 2, 3, 4, 0,\dots
\]
(\href{https://oeis.org/A014586}{sequence A014586} in the Online Encyclopedia of Integer Sequences, OEIS) displays no obvious patterns.

Subtract-a-square has another reason for interest, beyond investigations related to combinatorial game theory. The set $C$ of cold positions in this game,
\[
0,2,5,7,10,12,15,17,20,22,34,39,44, 52, 57, 62, 65, 67, 72, 85, 95,\dots
\]
(\href{https://oeis.org/A030193}{sequence A030193} in the OEIS) has the property that no two elements of $C$ differ by a square number. A sequence with this property is called a square-difference-free set. The cold positions of subtract-a-square are square-difference-free because, whenever $c$ is a cold position, and $i$ is a positive integer, $c+i^2$ must be hot, as one could win by moving from $c+i^2$ to $c$.
 The square-difference-free sets have been extensively investigated in number theory, following the work of Furstenberg~\cite{Fur-JAM-77} and S\'ark\"ozy~\cite{Sar-AMASH-78}, who showed that they have natural density zero. This means that, for all $\epsilon$, there exists an $N$ such that, for all $n>N$, the fraction of positive integers up to $n$ that belong to the set is at most $\epsilon$.
 
More strongly, the set $C$ of cold positions in subtract-a-square is a maximal square-difference-free set. Every positive integer that is not in $C$ (a hot position) has a move to a cold position, so it could not be added to $C$ without destroying the square-difference-free property. Every maximal square-difference-free subset of the range $[0,n]$ must have size at least $\Omega(\sqrt{n})$ (otherwise there would not be enough sums or differences of set elements and squares to prevent the addition of another number in this range)\footnote{See Golomb~\cite{Gol-JCT-66}, Theorem 4.1.} and size at most
 \[
 O\left(\frac{n}{(\log n)^{\frac{1}{4}\log\log\log\log n}}\right)
 \]
by quantitative forms of the Furstenberg--S\'ark\"ozy theorem~\cite{PinSteSze-JLMS-88}. In particular, these bounds apply to $|C\cap[0,n]|$, the number of cold positions of subtract-a-square up to~$n$.  However, it is not known whether these upper and lower bounds are tight or where the number of cold positions lies with respect to them. In the densest known maximal square-difference-free sets, the number of elements up to $n$ is
 \[
 \Omega\left(n^{(1+\log_{205}12)/2}\right)\approx n^{0.733412}.
 \]
 The construction for these dense sets involves finding a square-difference-free set 
modulo some base~$b$, and selecting the numbers whose base-$b$ representation has these values in its even digit positions and arbitrary values in its odd digit positions~\cite{Ruz-PMH-84}. The bound given in the formula above comes from applying this method to a square-difference-free set of 12 values modulo 205~\cite{BeiGas-08,Lew-EJC-15}. Plausibly, a greater understanding of the nim-values of subtract-a-square could lead to progress in this area of number theory.

Algorithmically, for a subtraction game in which the allowed moves are to take a number of tokens in a given set $S$, the nim-values can be computed by dynamic programming, using the recurrence
\[
\operatorname{nimvalue}(n)=\operatorname*{mex}_{i\in S, i\le n} \operatorname{nimvalue}(n-i).
\]
Here, the ``mex'' operator (short for ``minimum excludent''~\cite{ONAG-11}) returns the smallest non-negative integer that cannot be represented by an expression of the given form. No separate base case is needed, because in the base case (when $n=0$), the set of available moves (numbers in $S$ that are at most $n$) is empty and the mex of an empty set of choices is zero. Evaluating this recurrence, for all heap sizes up to a given threshold $n$,
takes time $O(n|S|)$. The set $C$ of cold positions can be determined within the same time bound,
by applying this recurrence and then returning the set of positions whose nim-value is zero.

However, in the study of algorithms, many naive dynamic programming algorithms turn out to be suboptimal: they can be improved by more sophisticated algorithms for the same problem. Is that the case for this one? We will see that it is. We provide the following two results:
\begin{itemize}
\item We show how to compute the set of cold positions in a given subtraction game, for heaps of size up to a given threshold $n$, in time $\tilde O(n)$.
\item We show how to compute the nim-values of a given subtraction game, 
for heaps of size up to a given threshold $n$, in time $\tilde O(mn)$.
\end{itemize}
In these time bounds, the $\tilde O$ notation elides logarithmic factors in the time bound, and the parameter $m$ refers to the maximum nim-value of any position within the given range.

Ignoring the logarithmic factors hidden in the $\tilde O$ notation, our time bounds are always at least as good as the $O(n|S|)$ time for naive dynamic programming, because for any subtraction game $m\le |S|$ (if there are only $|S|$ possible moves, the mex of their values can be at most $m$).
But are they actually a significant improvement? To answer this, we need to know how quickly $m$ grows compared to the known growth rate of $|S|$.

To determine whether our algorithms provide a speedup for the game of subtract-a-square, we performed a sequence of computational experiments to determine the density of this game's cold positions and the growth rate of its largest nim-values. We find experimentally that, up to a given $n$, the largest nim-value appears to grow as $O(n^{0.35})$, significantly more slowly than the $O(n^{1/2})$ growth rate of the subtraction set. The difference in the growth rates for these quantities shows that our algorithms are indeed an asymptotic improvement by a polynomial factor.
Additionally, the number of cold positions appears to grow at least as quickly as $n^{0.69}$.
That is, the cold positions of this game provide an unexpectedly large square-difference-free set,
competitive with the best theoretical constructions for these sets.
 Examining the modular structure of the set of cold positions, we find that it appears to be similar to the structure of these theoretical constructions, with a square-difference-free set of digit values in even positions and arbitrary values in odd positions. 

\section{Algorithms}

\subsection{Subtraction with hotspots}

In order to evaluate subtraction games efficiently, it will be convenient to generalize them somewhat, to a class of \emph{subtraction games with hotspots}. Given two sets $S$ and $H$ (of positive and non-negative integers respectively), we define a subtraction game with subtraction set $S$ and hotspot set $H$ as follows. The game starts with a single pile of some number of tokens, and the players alternate in choosing a number from $S$ and removing that number of tokens from the pile, as before. However, if any move leaves a pile whose remaining number of tokens belongs to $H$, then the player who made that move immediately loses. (This is not quite the same as allowing the other player to remove all the tokens from piles whose size belongs to $H$, because $H$ might contain the number zero, in which case removing all the tokens could be a losing move instead of a winning move.)

The presence of these hotspots makes defining a nim-value for these games problematic: they are not played by the normal winning convention, so what would happen if we played a game with multiple piles and one player moved to a hotspot? Nevertheless, a recurrence of the usual form suffices to determine the set of hot and cold positions of such a game:
\[
\operatorname{hot}(n)=n\in H \vee
\bigvee\limits_{i\in S, i\le n} \lnot\operatorname{hot}(n-i).
\]

\subsection{Finding the hotspots}

In a subtraction game (with subtraction set $S$, with or without hotspots), suppose that some set $C$ of positions has already been determined to be cold.
Then all positions $H$ that can reach $C$ in a single move are automatically hot.
We can formulate membership in this set of hot positions as a Boolean formula in conjunctive normal form (2-CNF):
\[
(i\in H) \Longleftrightarrow \bigvee\limits_{j+k=i} (j\in C)\wedge (k\in S).
\]

Now suppose that $C$ and $S$ are both represented as bitvectors: arrays of binary values that are $0$ for non-members and $1$ for members of each set.
Then the problem of computing the bitvector representation of~$H$ from the above formula is a standard problem known as \emph{Boolean convolution}, studied for its applications in string matching~\cite{FisPat-CC-74,MutPal-STOC-94,Kal-SODA-02}. It is an instance of a more general class of convolution problems in which we compute
\[
C[i]=\mathop{\oplus}\limits_{j+k=i}A[j]\otimes B[j]
\]
for an ``addition'' operation $\oplus$ and ``multiplication'' operation $\otimes$. In Boolean convolution, $\otimes$ is conjunction ($\wedge$), and $\oplus$ is disjunction ($\vee$).

If the input bitvectors have total length $n$, their Boolean convolution can be computed in $O(n\log n)$ time by replacing their Boolean values with the numbers $0$ and $1$ and computing a numerical convolution (with addition as $\oplus$ and multiplication as $\otimes$) using the fast Fourier transform algorithm.

\subsection{Divide and conquer}

We are now ready to describe our algorithm for finding the hot and cold positions of a subtraction game with hotspots. We assume that we are given as input a range $[x,y)$ of integer values to evaluate (following the Python convention for half-open integer ranges where the bottom delimiter is inside the range and the top delimiter is outside it), together with two sets: the subtraction set $S$ and a set $H$ of predetermined hotspots.

As a base case, if the range has zero or one values in it, we can solve the problem directly: each value in the range is hot or cold accordingly as it belongs or does not belong to $H$, respectively. Otherwise, we perform the following steps:
\begin{enumerate}
\item Find the midpoint $m=(x+y)/2$ of the range, and partition the range into the two subranges
$[x,m)$ and $[m,y)$.
\item Recursively evaluate the lower subrange $[x,m)$, determining its hot and cold positions ($H_x$ and $C_x$ respectively).
\item Use Boolean convolution to find the positions $H_m$ in the upper subrange $[m,y)$
that are hot because they can be reached in a single step from a cold position $C_x$ in the lower subrange.
\item Recursively evaluate the upper subrange $[m,y)$, with hotspot set $H\cup H_m$,
determining its hot and cold positions ($H_y$ and $C_y$ respectively).
\item Return the hot set $H_x\cup H_y$ and cold set $C_x\cup C_y$.
\end{enumerate}

The time for this algorithm can be analyzed using the master method, as is standard for such divide-and-conquer algorithms, giving the following result:

\begin{theorem}
We can determine which positions are hot and which are cold, in a range of $n$ positions of a subtraction game with hotspots, in time $O(n\log^2 n)$.
\end{theorem}

\subsection{Nim-values}

We can reduce the computation of nim-values in a subtraction game to the computation of hot and cold positions in a subtraction game with hotspots, via the following lemma.

\begin{lemma}
\label{lem:hot-from-cold}
Let $S$ be a subtraction set, and let $H$ be the set of positions in the subtraction game for $S$ that have nim-value at most $t$. Then the positions that have nim-value $t+1$ are exactly the cold positions of the subtraction game with hotspots with subtraction set $S$ and hotspot set $H$.
\end{lemma}

\begin{proof}
A position has nim-value $t+1$ if it does not belong to $H$ (else it would have a smaller nim-value) and does not have a move to a smaller position with nim-value $t+1$ (else $t+1$ would not be one of its excluded values). But this is exactly the defining condition for the cold positions of the subtraction game with hotspots.
\end{proof}

\begin{theorem}
In any subtraction game with subtraction set $S$, we can determine the nim-values of the first $n$ positions in time $O(mn\log^2 n)$, where $m$ is the maximum nim-value of any of these positions.
\end{theorem}

\begin{proof}
We loop over the range of nim-values from $0$ to $s$, using
\autoref{lem:hot-from-cold} to compute the set of positions having each successive nim-value in time $O(n\log^2 n)$ per nim-value.  The loop terminates when all of the first $n$ positions have been assigned a nim-value.
\end{proof}

The maximum nim-value of a subtraction game is $|S|$, so (except for the logarithmic factors) this time bound compares favorably with a naive $O(n|S|)$ dynamic programming algorithm for computing the nim-values of each position by finding the minimum excluded value among the other positions reachable from it.

\section{Experiments}

To compare the performance of our Boolean convolution based evaluation algorithms to naive algorithms for subtract-a-square, we performed some computational experiments, which we describe here.

\subsection{Maximum nim-value}

\begin{figure}[t]
\centering\includegraphics[width=0.85\textwidth]{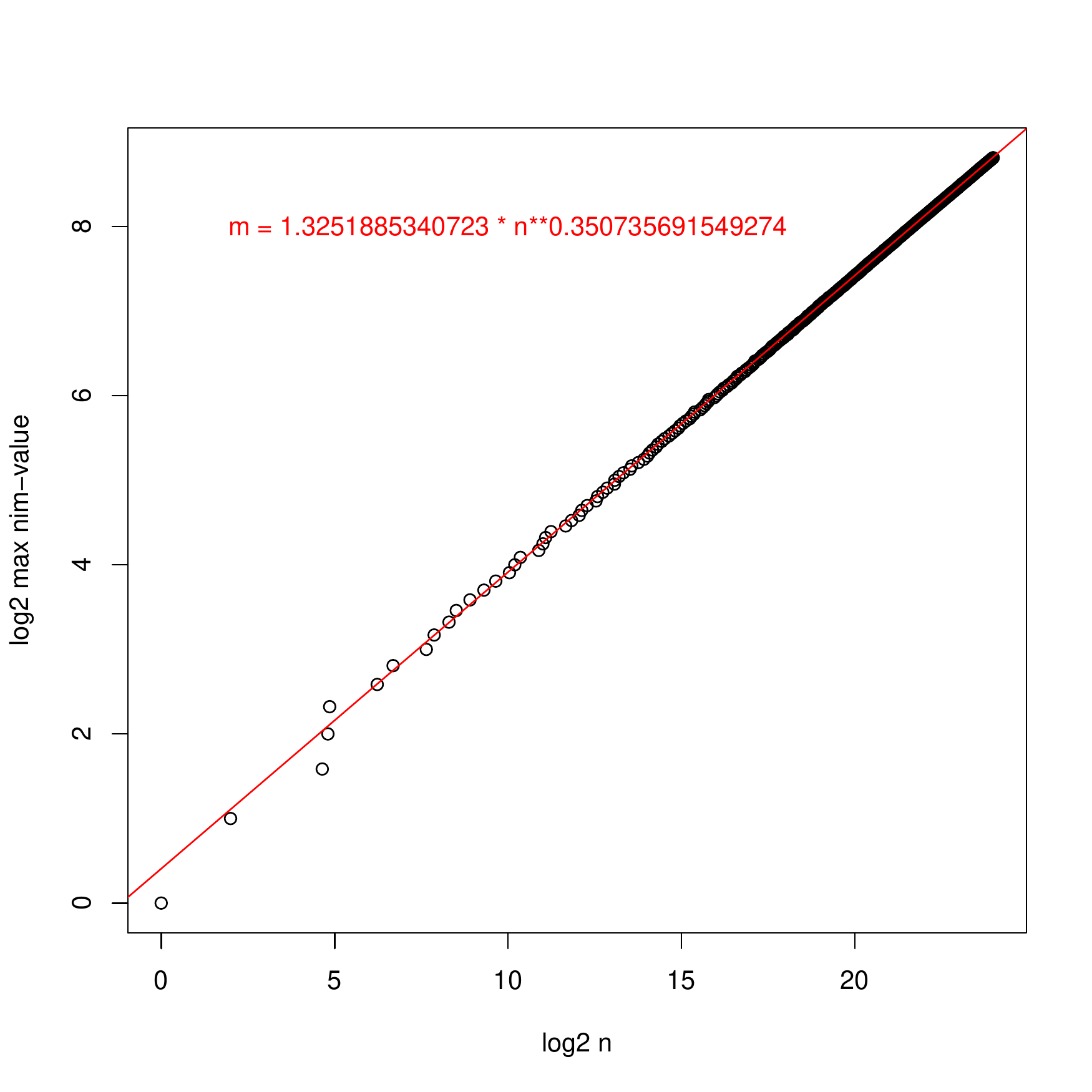}
\caption{The maximum nim-values $m$ seen among the first $n$ positions in subtract-a-square.}
\label{fig:sas}
\end{figure}

\autoref{fig:sas} plots (on a doubly logarithmic scale) the maximum nim-values $m$ seen among the first $n$ positions in subtract-a-square. Only the positions where a new maximum is attained are included in the plot.

We fitted a function of the form $cn^e$ (a monomial with constant coefficient $c$ and exponent $e$) to these points, by using Siegel's repeated median estimator~\cite{Sie-BM-82}, a form of robust statistical regression that is insensitive to outliers (as would be expected to occur in the lower left parts of the plot).
This estimator fits a line through a sample of points by, for each point, computing the median of the slopes formed by it and the other points, and then choosing the slope of the fit line to be the median of these medians. It similarly chooses the height of the fit line so that it passes above and below an equal number of points. We applied this to the points on our log-log plot, using the mblm library of the R statistical package, which implements this estimator, and then transformed the fit line back to a monomial over the original coordinates of the data points. The result is shown in red in the figure.

As the figure shows, the maximum nim-value $m$ among the first $n$ positions of subtract-a-square is accurately estimated by a function of the form $O(n^{0.351})$, well below the $O(n^{0.5})$ size of the subtraction set for this game. Therefore, we would expect our $O(mn\log^2 n)$-time convolution-based algorithm for computing the nim-values of this game to be asymptotically faster than the $O(n^{3/2})$ time for dynamic programming. However, even if we ignore the different constant factors in the running times of these two algorithms, $n$ needs to be approximately $10^{26}$ in order for $n^{1.35}\log_2^2n$ to be smaller than $n^{1.5}$, so we would not expect this speedup to be applicable to practically relevant ranges of $n$. Because of the simplicity and relative efficiency of the dynamic programming algorithm for small $n$, the values in \autoref{fig:sas} (for $n$ up to $2^{24}$) were computed by dynamic programming rather than convolution.

\subsection{Number of cold positions}

\begin{figure}[t]
\centering\includegraphics[width=0.85\textwidth]{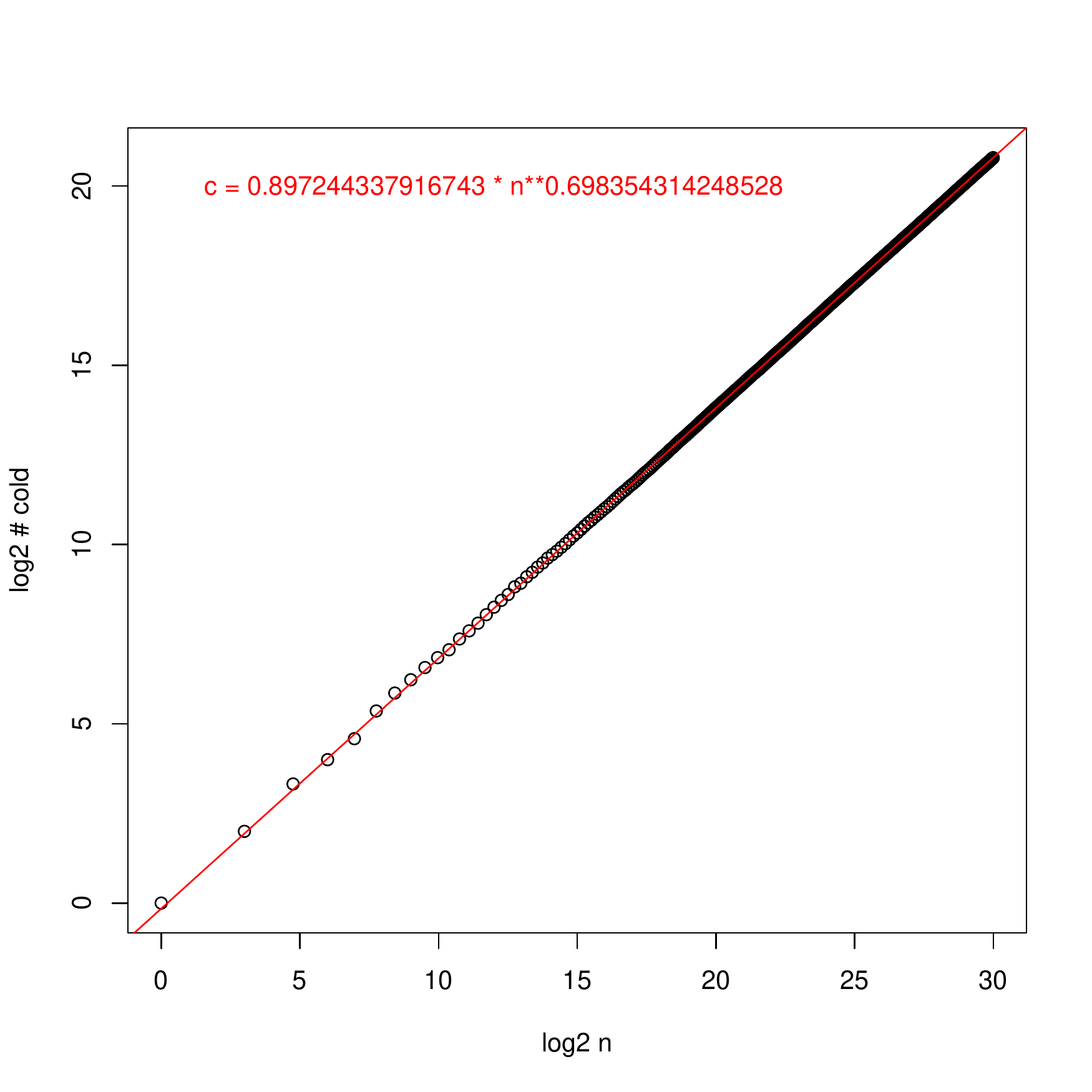}
\caption{The number of cold positions among the first $n$ positions in subtract-a-square.}
\label{fig:density}
\end{figure}

Our next experiment measures the number of cold positions among the first $n$ positions in subtract-a-square (\autoref{fig:density}). In order to provide more data points in the lower left part of the log-log plot than would be visible if we used uniform sampling of the range of values of~$n$,
the plot of \autoref{fig:density} shows the number of cold positions for each value of~$n$ that is a perfect cube (that is, for the values $1,8,27,64,\dots$) up to $2^{30}$.
As in the previous experiment, we fitted a monomial function to these points using Siegel's repeated median estimator.

The number of cold positions does not directly affect the time bound for our convolution-based algorithm. However, it does affect the time for a different algorithm, for computing the set of cold positions (but not their nim-values) directly, in any subtraction game. This algorithm is analogous to the sieve of Eratosthenes, which finds prime numbers iteratively, for each one marking off the numbers that are not prime. To compute the cold positions among the first $n$ positions, it performs the following steps:
\begin{enumerate}
\item Initialize a Boolean array $H$ of length $n$ (indicating whether each position is hot) to be false in each cell.
\item For each position $i$ from $0$ to $n$, test whether $H[i]$ is still false. If it is, perform the following steps:
\begin{enumerate}
\item Output $i$ as one of the cold positions.
\item For each value $s$ in the subtraction set $S$, mark $i+s$ as hot by setting $H[i+s]$ to be true.
\end{enumerate}
\end{enumerate}
If the set of cold positions up to $n$ is $C$, and the subtraction set is $S$, then this sieving algorithm takes time $O(|C|\cdot|S|)$.

In some subtraction games, $C$ could be as small as $n/|S|$, in which case the sieving algorithm would take linear time. However, our experiments show that, for subtract-a-square, $C$ appears to grow more like $n^{0.7}$, giving the sieving algorithm a running time of approximately $n^{1.2}$, compared to the $O(n\log^2 n)$ time bound of the convolution based algorithm. Again ignoring the constant factors in the time bounds, $n$ would need to be approximately $10^{18}$ for the convolution-based algorithm to be faster than the sieving algorithm. Because it is simple to code and fast for smaller values of~$n$, the results in \autoref{fig:density} were calculated using the sieving algorithm.

\subsection{Modularity}

\begin{figure}[t]
\centering\includegraphics[width=0.85\textwidth]{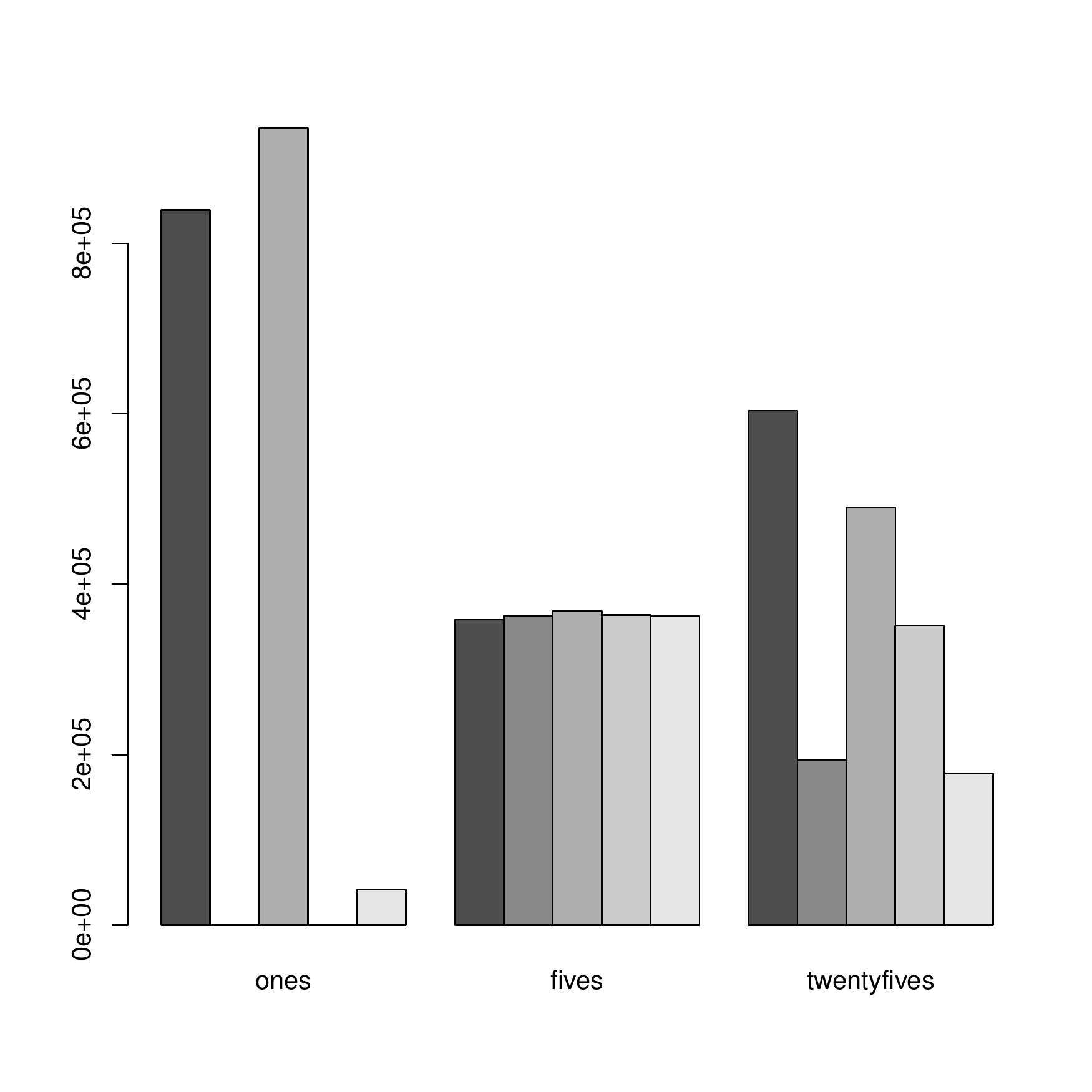}
\caption{The distribution of digit values among the three low-order base-5 digits of cold positions
(for $n<2^{30}$) in subtract-a-square.}
\label{fig:mod5}
\end{figure}

\begin{figure}[t]
\centering\includegraphics[width=0.85\textwidth]{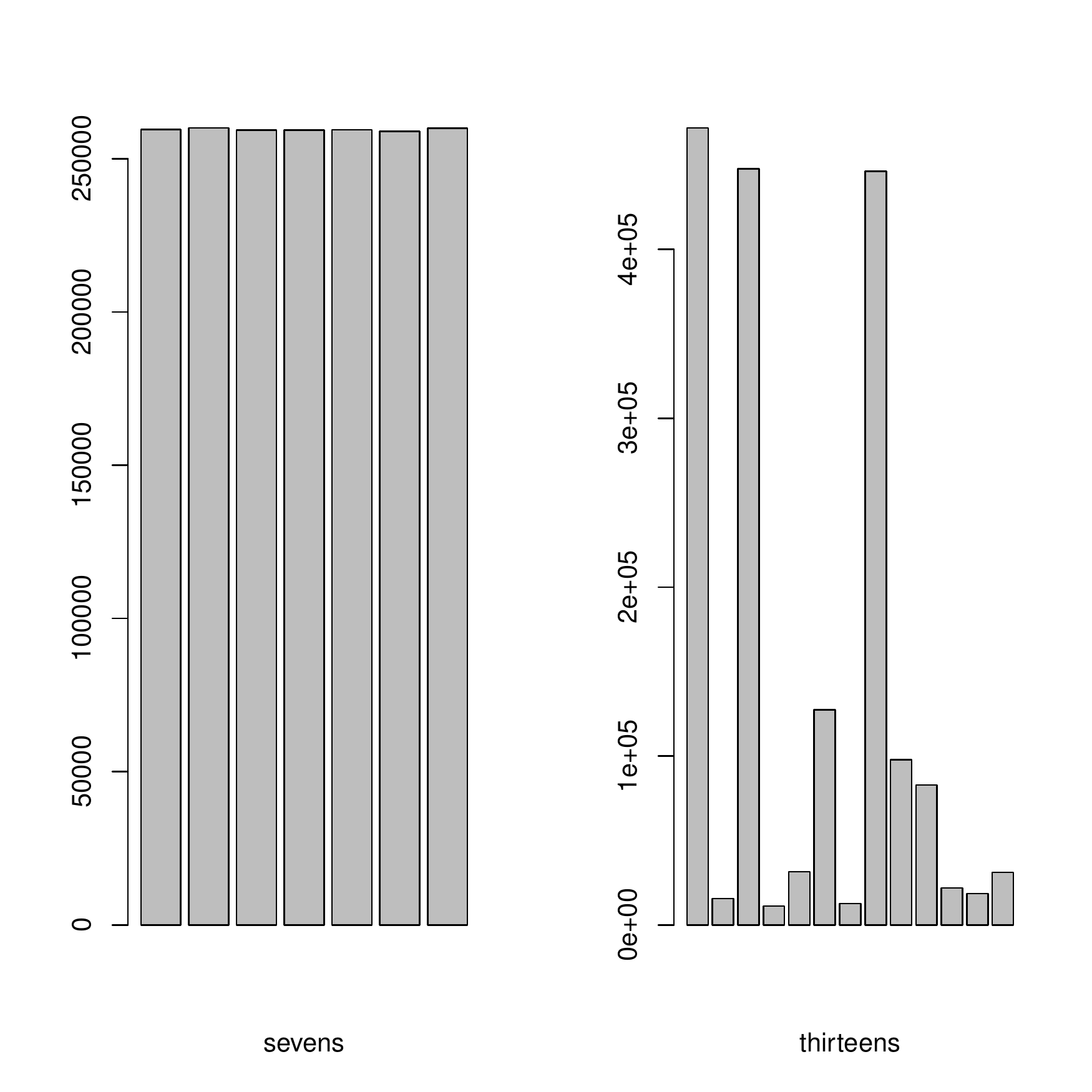}
\caption{The distribution of digit values among the low-order base-7 and base-13 digits of cold positions
(for $n<2^{30}$) in subtract-a-square.}
\label{fig:moremods}
\end{figure}

The high density of cold positions in subtract-a-square is surprising, especially in view of earlier conjectures in the theory of square-difference-free sets that the number of values up to $n$ in such a set could be at most $n^{1/2+o(1)}$~\cite{Sar-AUSB-78}. These conjectures were disproven by finding sets of numbers of a special form: numbers whose radix-$b$ representation, for a carefully chosen base $b$, use only base-$b$ digits from a square-difference-free set (modulo~$b$) in their even digit positions~\cite{Ruz-PMH-84,BeiGas-08,Lew-EJC-15}. Although the cold positions of subtract-a-square have somewhat lower density than these constructions, they arise more naturally, and it is of interest to investigate their modular structure and compare it to the structure of these other known dense square-difference-free sets.

The idea of considering the base-$b$ structure of these positions, for different choices of the base~$b$, also arises from the consideration of a different subtraction game, described by Golomb~\cite{Gol-JCT-66}. This game has as its subtraction set the Moser--de Bruijn sequence
\[
0, 1, 4, 5, 16, 17, 20, 21, 64, 65, 68, 69,\dots
\]
of numbers that are sums of distinct powers of four. That is, when written in base~4, the numbers of the subtraction set have only 0 and 1 as their base-4 digits. The nim-value of any position $n$ may be obtained by writing $n$ in base~4, taking each digit modulo~2 (reducing it to 0 or 1), and then reinterpreting the resulting string of 0's and 1's as a binary number. Because of this simple formula for its nim-values, the Moser--de Bruijn subtraction game has both a maximum nim-value and a number of cold positions (among the first~$n$ positions) proportional to $\Theta(\sqrt n)$. It subtraction set size, also $\Theta(\sqrt n)$, is comparable to that for subtract-a-square.
 In particular, for this game, convolution is neither asymptotically faster than dynamic programming nor than sieving, although all of these algorithms can be improved by using the formula instead. What makes subtract-a-square so different from the Moser--de Bruijn subtraction game?

To approach these questions, we performed more computational experiments studying the distribution of digit values for the cold positions in subtract-a-square, for various bases. This study follows the earlier work of Golomb~\cite{Gol-JCT-66}, who observed that the low-order base-5 digits of the cold positions among the first first 20,000 game positions were highly non-uniformly distributed, and of Bush~\cite{Bus-sm-92}, who extended this study to the first 40,000,000 game positions. \autoref{fig:mod5} shows an extension of this study to the first $2^{30}$ game positions,
and to the three low-order base-5 digits of each cold position. As the figure shows, with a few exceptions, the ones digit of the cold positions lies within the square-difference-free set $\{0,2\}\pmod{5}$. The fives digit shows no significant non-uniformities, but the twentyfives digit is quite non-uniformly distributed, and is possibly heading towards the same square-difference-free set $\{0,2\}\pmod{5}$. In this way, the cold positions of subtract-a-square appear to be emulating the strategy of the known dense square-difference-free sets~\cite{Ruz-PMH-84,BeiGas-08,Lew-EJC-15} of having a square-difference-free set of digits in even digit positions and all possible digits in odd positions modulo a base~$b$. In this case $b=5$, and following this strategy perfectly for $b=5$ would lead to a set of size $n^{\log_{10} 25}\approx n^{0.71534}$. The slightly slower growth rate of the cold positions in subtract-a-square can be explained by the slow convergence of its higher-order base-5 digits to square-difference-free sets of digits.

What about other bases? \autoref{fig:moremods} shows the results of the same experiment (for the low-order digits only) for base~7 and base~13. Because 7 is 3 modulo 4, there are no nontrivial square-difference-free sets modulo 7: every two numbers modulo~7  differ by a square (mod~7).
Perhaps because of this, the digit values in base 7 show no significant nonuniformities. However, modulo 13, the squares are $0$, $\pm 1$, $\pm 3$, and $\pm 4$. Because 13 is 1 modulo~4, each of the nonzero squares occurs four times among the squares of values mod~13; for instance, $\pm 3$ is the square of $4$, $6$, $7$,  and $9$ (mod~13). As the figure shows, the low-order digits of the base-13 representations of the cold positions in subtract-a-square appear to be converging towards the square-difference-free set $\{0,2,7\}$ (mod~13). Perhaps subtract-a-square implements the modular strategy for finding dense square-difference-free sets in all prime bases (congruent to 1~mod~4) simultaneously?

\section{Conclusions}

We have developed new convolution-based methods for evaluating arbitrary subtraction games (either to determine the set of cold positions or to evaluate the nim-value of each position).
Our experiments on the subtract-a-square game show that its maximum nim-value is lower than the theoretical value for games with subtraction sets of the same size, and its number of cold positions is higher than the theoretical value. These results show that, asymptotically, our new algorithms are faster than alternative dynamic programming or sieving approaches for the same problems on this game. However, the breakeven point for the new algorithms is high enough that our convolution-based approach is not yet practical. It would be of interest to develop improved algorithms that are both asymptotically faster and more practical than existing approaches.

In an attempt to investigate why the cold positions of subtract-a-square produce a dense square-difference-free set, we investigated the base-$b$ representations of the cold positions for several small prime choices of $b$. Our tests found significant irregularities in the even positions of these base-$b$ representations, when $b$ is congruent to 1 mod 4. We leave the problem of finding a theoretical explanation for these patterns, and for the density of the cold positions in subtract-a-square, as open for future research.

\bibliographystyle{abuser}
\bibliography{subtraction}
\end{document}